\documentclass[aps,prl,twocolumn,superscriptaddress, amsmath, amssymb, amsfonts, nofootinbib]{revtex4-1}
\def\draft{1}
\usepackage{amsthm}
\usepackage{color}

\usepackage{ifthen}
\newcommand{\YSHI}[1]{\ifthenelse{\equal{\draft}{1}}{{\color{red}{#1}}}{#1}}

\newtheorem{theorem}{Theorem}
\newtheorem{proposition}[theorem]{Proposition}
\newtheorem{definition}[theorem]{Definition}
\newtheorem{lemma}[theorem]{Lemma}
\newtheorem{corollary}[theorem]{Corollary}

\usepackage{graphics}

\newcommand{\commentout}[1]{}

\def\Tr{\textnormal{Tr}}

\def\Supp{\textnormal{Supp }}

\def\Dist{\textnormal{Dist}}

\begin{document}

\title{Randomness in nonlocal games between mistrustful players}

\author{Carl A.~Miller}
\email{camiller@umd.edu}
\affiliation{National Institute of Standards and Technology, 100 Bureau Dr., Gaithersburg, MD  20890, USA, \\
Joint Center for Quantum Information and Computer Science, University of Maryland, College Park, MD  20742, USA}

\author{Yaoyun Shi}
\email{shiyy@umich.edu}
\affiliation{Dept.~of Electrical Engineering and Computer Science, University of Michigan, Ann Arbor, MI  48109, USA}

\date{\today}

\begin{abstract}
\noindent
If two quantum players at a nonlocal game $G$ achieve a superclassical score, then their measurement outcomes
must be at least partially random from the perspective of any third player.  This is the basis
for device-independent quantum cryptography.  In this paper we address a related question: does a superclassical
score at $G$ guarantee that one player has created randomness from the perspective of the other player?  We show
that for complete-support games, the answer is yes: even if the second player is given the first player's input at the
conclusion of the game, he cannot perfectly recover her output.  Thus some amount of \textit{local} randomness (i.e., randomness
possessed by only one player) is always obtained when randomness is certified from nonlocal games with quantum strategies. 
This is in contrast to non-signaling game strategies, which may produce global randomness without any local randomness.
We discuss potential implications for cryptographic protocols between mistrustful parties.
\end{abstract}


\maketitle

\section{Introduction}

When two quantum parties Alice and Bob play a nonlocal game $G$ and achieve
a score that exceeds the best classical score $\omega_c ( G )$, their outputs must be at least partially random.
In other words, all Bell inequality violations certify the existence of randomness.  
This fact is at the center of protocols for device-independent
quantum cryptography, where untrusted devices are used to perform cryptographic procedures.
In particular, this notion of certification is the basis for device-independent \textit{randomness expansion}, where a small random seed is converted into a much larger uniformly random output by repeating Bell violations  \cite{col:2006, pam:2010, ColbeckK:2011,
Vazirani:dice, pironio13, fehr13, CoudronVY:2013, CY:STOC, Miller:2016,
Universal-spot, Dupuis:2016, Arnon:2016}.

A natural question arises:
is new randomness also generated by one player from the perspective of the other player?  Specifically, 
if $X$ denotes Alice's outputs, $Z$ denotes the post-measurement state that Bob has at the conclusion of the game, and $F$ denotes all side information (including
Alice's input), is there a certified lower bound for the conditional
entropy $H ( X \mid Z F )$?  Besides helping us understand the nature of certified randomness, this particular
kind of randomness (local randomness) has applications in mutually mistrustful cryptographic settings, where
Alice and Bob are cooperating but have different interests.

Quantifying local randomness (i.e., randomness that is only known to one player) is challenging because many of the known tools do not apply.
Lower bounds for the total randomness (i.e, randomness from the perspective of
an outside adversary) have been computed as
a function of the degree of the Bell violation (see Figure 2 in \cite{pam:2010}) but 
are not directly useful for certifying local randomness.  One of the central challenges
is that we are measuring randomness from the perspective of an active, rather than passive, adversary:
Bob's guess at Alice's output occurs after Bob has carried out his part of the strategy for $G$.  
Current tools for device-independent randomness expansion are not designed to address the case
where the adversary is a participant in the nonlocal game.

Does the generation of certified randomness always involve the generation
of \textit{local} certified randomness?  The answer
is not obvious: for example, in the non-signaling setting, Alice and Bob could share a PR-box\footnote{That is,
the unique $2$-part non-signaling resource whose input bits $a,b$ and output bits $x,y$ always
satisfy $x \oplus y = a \wedge b$} which generates $1$ bit of certified randomness per use, but
no new local randomness -- Bob could perfectly guess Alice's output from his own if he were given
Alice's input. 

Motivated by the above, we prove the following result in this paper (see Theorem~\ref{robustthm} for a formal
statement).
\begin{theorem}[Informal]
\label{robustthminformal}
For any complete-support game\footnote{That is, a game in which each
input pair occurs with nonzero probability.} $G$, there is a constant $C_G > 0$ such that the following holds.
Suppose Alice and Bob use a strategy for $G$ which achieves a score that is $\delta$ above the best classical
score (with $\delta > 0$).  Then, at the conclusion of the strategy and given Alice's input, Bob can guess
her output with probability at most $(1- \delta^2/C_G)$.
\end{theorem}

We note that similar problems have been studied in the literature in settings different from ours.
There has been other work examining the scenario where a third party
tries to guess Alice's output after a game
(e.g., \cite{Pawlowski:2010}, \cite{Kempe:2011b}, \cite{Acin:PRA:2016}), and single-round games
have appeared
where Bob is sometimes given
only Alice's input, and asked to produce her output
(e.g., \cite{Mancinska:2014}, \cite{Vazirani:QKD:PRL}, \cite{Vidick:2013}).  (We believe the novelty of our scenario
in comparison to these papers is that we consider the randomness of Alice's output
\textit{after} Bob has performed his part of a quantum strategy, and thus has potentially lost information due to measurement.)
Two recent papers
also address randomness between multiple players, under assumptions about imperfect storage
\cite{Kaniewski:2016, Rib:2016}.

In addition to the above, we prove a structural theorem for quantum strategies
that allow perfect guessing by Bob.  Not only do such strategies
not achieve Bell inequalities, but they are also {\em essentially classical} in the following sense.
Let $D, E$ denote the quantum systems possessed by Alice and Bob, respectively
\begin{theorem}[Informal]
\label{informalthm}
Suppose that Alice's and Bob's strategy is such that if the game $G$ is played and then Bob is
given Alice's input, he can perfectly guess her output.  Then,
there is an isometry mapping Bob's system to $E_1 \otimes E_2$ such that Bob's strategy for $G$ involves only $E_1$, and all of Alice's observables commute with the reduced state on $DE_1$.
\end{theorem}
(See Theorem~\ref{thm:main} and Corollary~\ref{maincor} for a formal statement.)
Thus, in the case of perfect guessing, the strategy is equivalent to one in which Alice's measurements have no effect
on the shared state.

\subsection{Structure of the paper}

We begin with the case of perfect guessing.  We formalize the concept of
an essentially classical strategy, using a definition of equivalence between
strategies which is similar to definitions used in results on quantum rigidity.
We then give the proof of Theorem~\ref{informalthm}.  It is known that two sets of mutually commuting measurements on
a finite-dimensional space can be expressed
as the pullback of bipartite measurements. This fact is used along with matrix algebra
arguments to show the necessary splitting
of Bob's system into $E_1 \otimes E_2$.

Then we proceed 
with the proof of Theorem~\ref{robustthminformal}.  
It has been observed by previous work (e.g., \cite{Mancinska:2014}, \cite{vidickthesis}) that if a measurement $\{ P_i \}$
on a system $D$ from bipartite state $\rho_{DE}$ are highly predictable via measurements on $E$, then
the measurement does not disturb the reduced state by much: $\sum_i P_i \rho_D P_i \sim \rho_D$.  In this paper we give
a simplified proof of that fact (Proposition 11).  The interesting
consequence for our purpose is that if Alice's measurements are highly predictable to Bob, then Alice can
copy out her measurement outcomes in advance, thus making her strategy approximately classical.  We take
this a step further, and show that if Bob first performs his own measurement on $E$ the resulting classical-quantum
correlation is also approximately preserved by Alice's measurements (which is not necessarily true of the original entangled state $\rho_{AB}$).
This is sufficient to  show that an approximately-guessable strategy yields an approximately classical strategy.

The subtleties in the proof are in establishing the error terms that arise when Alice copies out multiple measures
from her side of the state.  We note that the proof crucially requires that the game $G$ has complete support.  An interesting
further avenue is to explore how local randomness may break down if the condition is not satisfied.

In section~\ref{sec:disc} we discuss the implications of our result.

\section{Preliminaries}

For any finite-dimensional Hilbert space $V$, let
$L ( V )$ denote the vector space of linear
automorphisms of $V$.  For any $M, N \in L ( V )$,
we let $\left< M , N \right>$ denote $\Tr [ M^* N ]$.
If $S \subseteq V$ is a subspace of $V$,
let $\mathbf{P}_S \in L ( V )$ denote orthogonal projection
onto $V$.

Throughout this paper we fix four disjoint finite sets
$\mathcal{A, B, X, Y}$, which denote, respectively,
the first player's input alphabet, the second player's
input alphabet, the first player's output alphabet,
and the second player's output alphabet.
A \textit{$2$-player (input-output) correlation} is a vector
$(p_{ab}^{xy})$ of nonnegative reals,
indexed by $a, b, x, y \in \mathcal{A} \times \mathcal{B}
\times \mathcal{X} \times \mathcal{Y}$, satisfying $\sum_{xy} p_{ab}^{xy}  = 1$
for all pairs $(a, b)$, and satisfying the condition that the quantities
\begin{eqnarray}
\begin{array}{ccc}
p_a^x := \sum_y p_{ab}^{xy}, & \hskip0.2in &
p_b^y := \sum_x p_{ab}^{xy}
\end{array}
\end{eqnarray}
are independent of $b$ and $a$, respectively (no-signaling).

A $2$-player game is a pair $(q, H)$ where
\begin{eqnarray}
q \colon \mathcal{A} \times \mathcal{B} \to [0, 1 ]
\end{eqnarray}
is a probability distribution and
\begin{eqnarray}
H \colon \mathcal{A} \times \mathcal{B} \times \mathcal{X}
\times \mathcal{Y} \to [0, 1]
\end{eqnarray}
is a function. If $q(a, b)\ne 0$ for all $a\in \mathcal{A}$ and $b\in\mathcal{B}$,
the game is said to have a {\em complete support}. The expected score associated to
such a game for a $2$-player correlation $(p_{ab}^{xy})$
is
\begin{eqnarray}
\sum_{a, b, x, y }  q ( a, b ) H ( a, b, x, y ) p_{ab}^{xy}.
\end{eqnarray}
We will extend notation by writing $q(a) = \sum_b q(a,b)$, $q(b) = \sum_a q(a,b)$,
and $q ( a \mid b ) = q(a, b) / q ( b)$ (if $q( b ) \neq 0$).

A \textit{$2$-player strategy} is a $5$-tuple
\begin{eqnarray}\label{eqn:strategy}
\Gamma & = & ( D, E, \{ \{ R_a^x \}_x \}_a , 
\{ \{ S_b^y \}_y \}_b , \gamma )
\end{eqnarray}
such that $D, E$ are finite dimensional Hilbert spaces,
$\{ \{ R_a^x \}_x \}_a$ is a family of $\mathcal{X}$-valued
positive operator valued measures (POVMs) on $D$ (indexed by $\mathcal{A}$),
$\{ \{ S_b^y \}_y \}_b$ is a family of $\mathcal{Y}$-valued
positive operator valued measures on $E$,
and $\gamma$ is a density operator on $D \otimes E$.
The \textit{second player states} $\rho_{ab}^{xy}$ of $\Gamma$
are defined by
\begin{eqnarray}
\rho_{ab}^{xy} & := & \Tr_D \left[ \sqrt{R_a^x \otimes S_b^y}
\gamma \sqrt{R_a^x \otimes S_b^y} \right]
\end{eqnarray}
(These states are, more explicitly, the subnormalized states of Bob's system
that arise after both Alice and Bob have performed their measurements.)
Define $\rho_a^x$ by the same expression with
$S_b^y$ replaced by the identity operator.  (These represent
the pre-measurement states of the second-player.)
Define $\rho:=\Tr_D(\gamma)
=\sum_x\rho_a^x$ for any $a$.

We say that the
strategy $\Gamma$ \textit{achieves} the $2$-player correlation
$(p_{ab}^{xy} )$ if $p_{ab}^{xy} =  \Tr [ \gamma ( R_a^x \otimes S_b^y ) ]$
for all $a, b, x, y$.
If a $2$-player correlation $(p_{ab}^{xy})$ can be achieved by
a $2$-player strategy then we say that it is a \textit{quantum}
correlation.

If $(p_{ab}^{xy})$ is a convex combination
of product distributions (i.e., distributions of the form
$(q_a^x ) \otimes (r_b^y )$ where $\sum_x q_a^x = 1$
and $\sum_y r_b^y = 1$) then we say that $(p_{ab}^{xy})$
is a \textit{classical} correlation.
Note that if the underlying state of a quantum strategy
is separable (i.e., it is a convex combination of bipartite
product states) then the correlation it achieves is classical.
The maximum expected score that can be achieved
for a game $G$ by a classical correlation is denoted $\omega_c ( G )$.

\section{Perfect Guessing}

We first address the case of perfect guessing --- that is, the case
when the second-player states $\{ \rho_{ab}^{xy} \}_x$ that remain
after the game is played are perfectly distinguishable by Bob.  It
turns out that this condition will imply some strong structural conditions
on the strategy used by Alice and Bob, and it will imply in particular
that Alice's and Bob's score at the game $G$ cannot be better
than that of any classical strategy.

\subsection{Congruent strategies}

It is necessary to identify pairs of strategies that
are essentially the same from an operational standpoint.
We use a definition that is similar to definitions
from quantum self-testing (e.g., Definition 2.13 in
\cite{McKagueBQP_published}).

A \textit{unitary embedding}
from a $2$-player strategy
\begin{eqnarray}
\Gamma & = & ( D, E, \{ \{ R_a^x \}_x \}_a , 
\{ \{ S_b^y \}_y \}_b , \gamma )
\end{eqnarray}
to another $2$-player strategy
\begin{eqnarray}
\overline{\Gamma} & = & ( \overline{D}, \overline{E}, \{ \{ \overline{R}_a^x \}_x \}_a , 
\{ \{ \overline{S}_b^y \}_y \}_b , \overline{\gamma} )
\end{eqnarray}
is a pair of unitary embeddings $i \colon D \hookrightarrow
\overline{D}$ and $j \colon E \hookrightarrow \overline{E}$
such that $\overline{\gamma} = (i \otimes j ) \gamma (i \otimes j)^*$,
$R_a^x = i^* \overline{R}_a^x i$, and $S_b^y
 = j^* \overline{S}_b^y j$.
 
 Additionally, if $\Gamma$ is such that $D = D_1 \otimes D_2$,
and $R_a^x = G_a^x \otimes I$ for all $a, x$, then we will call
the strategy given by
\begin{eqnarray}
( D_1, E, \{ \{ G_a^x \}_a \}_x , 
\{ \{ S_b^y \}_y \}_b , \Tr_{D_2} \gamma )
\end{eqnarray} 
a \textit{partial trace} of $\Gamma$.  We can similarly
define a partial trace on the second subspace $E$
if it is a tensor product space.

We will say that two strategies $\Gamma$ and $\Gamma'$
are \textit{congruent} if there exists a sequence
of strategies $\Gamma = \Gamma_1, \ldots, \Gamma_n =
\Gamma'$ such that for each $i \in \{ 1, \ldots, n-1 \}$,
either $\Gamma_{i+1}$ is a partial trace of
$\Gamma_i$, or vice versa, or there is a unitary embedding
of $\Gamma_i$ into $\Gamma_{i+1}$, or vice versa.
This is an equivalence relation.  Intuitively, two strategies
are congruent if one can be constructed from the other by adding
or dropping irrelevant information.
Note that if two strategies
are congruent then they achieve the same
correlation.

\subsection{Essentially classical strategies}
We are ready to define the key concept in this section and to
state formally our main theorem.

\begin{definition} A quantum strategy~(\ref{eqn:strategy}) is said to be {\em essentially classical}
if  it is congruent to one where $\gamma$ commutes with $R_a^x$ for all $x$ and $a$.
\end{definition}

Note that if the above condition holds, then applying the measurement
map
\begin{eqnarray}
X & \mapsto & \sum_x \sqrt{G_a^x} X \sqrt{G_a^x}
\end{eqnarray}
to the system $D$ leaves
the state $\gamma$ of $DE$ unchanged.

We are interested in strategies after the application of which Bob can predict Alice's output given her input.
This is formalized as follows. If $\chi_1, \ldots, \chi_n$ are positive semidefinite
operators on some finite dimensional Hilbert space
$V$, then we say that $\{ \chi_1, \ldots, \chi_n \}$
is \textit{perfectly distinguishable} if 
$\chi_i$ and $\chi_j$ have orthogonal support for
any $i \neq j$.  This is equivalent to the condition
that there exists a projective measurement on $V$ which
perfectly identifies the state from the
set $\{ \chi_1 , \ldots, \chi_n \}$.
\begin{definition} A quantum strategy~(\ref{eqn:strategy}) {\em allows perfect guessing} (by Bob) 
if for any $a, b, y$, $\{\rho_{ab}^{xy}\}_x$ is perfectly distinguishable.
\end{definition}

\begin{theorem}[Main Theorem] \label{thm:main} If a strategy for a complete-support game allows
perfect guessing, then it is essentially classical.
\end{theorem}

(We note that the converse of the statement is not true.
This is because even in a classical strategy, Alice's output may depend on some local randomness,
which Bob cannot perfectly predict.)

Before giving the proof of this result, we note the following proposition, which taken together with 
Theorem~\ref{thm:main} implies that any strategy that permits perfect guessing yields a classical correlation.

\begin{proposition}
\label{commutingprop}
The correlation achieved by an essentially classical strategy must be classical.
\end{proposition}

\proof{
We need only to consider the case that $\gamma$ commutes with
$R_a^x$ for all $a, x$.
For each $a \in \mathcal{A}$, let $V_a = \mathbb{C}^\mathcal{X}$,
and let $\Phi_a \colon L ( D ) \to L ( V_a \otimes D )$
be the nondestructive measurement defined by 
\begin{eqnarray}
\Phi_a ( T ) & = & \sum_{x \in \mathcal{X}} \left| x \right> \left< x \right|
\otimes \sqrt{R_a^x} T \sqrt{R_a^x}.
\end{eqnarray}
Note that by the commutativity
assumption, such operation leaves the state of $DE$ unchanged.

Since the measurements $\{R_a^x \}^x$ do not disturb the state of $DE$,
Alice can copy out all of her measurement outcomes in advance.
Without loss of generality, assume $\mathcal{A} = \{ 1, 2, \ldots, n \}$.
Let $\Lambda  \in L ( V_1 \otimes \ldots \otimes V_n \otimes D \otimes E )$
be the state that arises from applying the superoperators
$\Phi_1, \ldots, \Phi_n$, in order, to $\gamma$.  For any
$a \in \{ 1, \ldots, n \}$, the reduced state $\Lambda^{V_a E}$
is precisely the same as the result of taking the state
$\gamma$, applying the measurement $\{ R_a^x \}_x$
to $D$, and recording the result in $V_a$.
Alice and Bob can therefore generate the correlation $(p_{ab}^{xy})$
from the marginal state $\Lambda^{V_1 \cdots V_n E}$ alone (if
Alice possesses $V_1, \ldots, V_n$ and Bob possesses $E$).  Since this state is classical
on Alice's side, and therefore separable, the result follows.
}

\begin{corollary}\label{maincor}If a strategy for a complete-support game allows perfect guessing,
the correlation achieved must be classical. $\Box$
\end{corollary}

\subsection{Proving Theorem~\ref{thm:main}}
The proof will proceed as follows. 
First, we show that Alice's measurements $R_a := \{  R_a^x \}_x$ induce projective measurements
$Q_a:=\{Q_a^x\}_x$ on Bob's system.
Next, we argue that $Q_a$ commutes with Bob's own measurement $S_b:=\{S_b^y\}_y$ for any $b$.
This allows us to isometrically decompose Bob's system into two subsystems $E_1\otimes E_2$,
such that $S_b$ acts trivially on $E_2$, while $E_2$ alone can be used to predict $x$ given $a$.
The latter property allows us to arrive at the conclusion that $R_a$ commutes with $\gamma_{DE_1}$.

We will need the following lemma, which is well-known and commonly
attributed to Tsirelson.
We will only sketch the proof, and more details
can be found in Appendix A of \cite{Doherty:2008}.
The lemma asserts that for families of positive semidefinite
operators $\{ M_j \}, \{ N_k \}$ on a finite-dimensional
space $V$, \textit{commutativity} (i.e., the condition
that $M_j, N_k$ commute for and $j,k$) implies
\textit{bipartiteness} (i.e., the condition that $\{ M_j \}$ and $\{ N_k \}$
can be obtained as pullbacks via a map $V \to V_1 \otimes V_2$
of operators on $V_1$ and $V_2$, respectively).

\begin{lemma}
\label{commutelemma}
Let $\{ M_j \}, \{ N_k \}$ be positive semidefinite operators
on a finite-dimensional Hilbert space $V$ such that $M_j N_k =
N_k M_j$ for all $j,k$.  Then,
there exists a unitary embedding
$i \colon V \hookrightarrow V_1 \otimes V_2$ and
and positive semidefinite operators $\overline{M_j}$ on $V_1$
and $\overline{N}_k$ on $V_2$ 
such that
$M_j = i^* (\overline{M}_j \otimes \mathbb{I} ) i$ and
$N_k = i^* ( \mathbb{I} \otimes \overline{N}_k ) i$ for all $j, k, \ell, m$.
\end{lemma}

\textbf{Proof sketch.} Via the theory of von Neumann algebras,
there exists an isomorphism
\begin{eqnarray}
V & \cong & \bigoplus_\ell V_\ell \otimes W_\ell
\end{eqnarray}
under which
\begin{eqnarray}
M_j  & \cong & \bigoplus_\ell M_j^\ell \otimes I, \\
N_k & \cong & \bigoplus_\ell I \otimes N_k^\ell.
\end{eqnarray}
Let $V_1 = \bigoplus_\ell V_\ell, V_2 = \bigoplus_\ell W_\ell$,
and let $\overline{M}_j  =  \bigoplus_\ell M_j^\ell, \overline{N}_k  = \bigoplus_\ell N_j^\ell$. $\Box$

\begin{proof}[Proof of Theorem~\ref{thm:main}]
Express $\Gamma$ as in (\ref{eqn:strategy}).  
Without loss of generality, we may assume
that $\Supp \rho = E$.
By the assumption that $\Gamma$ allows perfect guessing,
for any $a$, the second-player states $\{ \rho_a^x \}_x$
must be perfectly distinguishable (since otherwise
the post-measurement states $\{ \rho_{ab}^{xy} \}_x$
would not be). 
Therefore, we can find projective measurements $\{ \{ Q_a^x \}_x \}_a$ on $E$
such that
\begin{eqnarray}
Q_a^x \rho Q_a^x = \rho_a^x.
\end{eqnarray}
Note that for any fixed $a$, if Alice and Bob were to prepare the state 
$\gamma_{DE}$ and Alice were to measure with $\{R_a^x \}_x$ and Bob were to measure with
$\{Q_a^x \}_x$, their outcomes would be the same.

We have that the states
\begin{eqnarray}
\label{rho1st} \rho_{ab}^{xy} & = & \sqrt{ S_b^y } Q_a^x \rho Q_a^x \sqrt{S_b^y } \\
\label{rho2nd}
\rho_{ab}^{x'y} & = & \sqrt{ S_b^{y} } Q_a^{x'} \rho Q_a^{x'} \sqrt{S_b^y }
\end{eqnarray}
have orthogonal support for any $x \neq x'$.  Since $\Supp \rho = E$,  we have $c \mathbb{I} \leq \rho$ for some $c > 0$.
Therefore,
\begin{eqnarray}
\left< \sqrt{ S_b^y } c Q_a^x \sqrt{S_b^y } , \sqrt{ S_b^y } c Q_a^{x'} \sqrt{S_b^y } \right> & = & 0,
\end{eqnarray}
which implies, using the cyclicity of the trace function,
\begin{eqnarray}
\left\| Q_a^x S_b^y Q_a^{x'} \right\|_2 & = & 0.
\end{eqnarray}
Therefore,
the measurements $\{ Q_a^x \}_x$ and $\{ S_b^y \}_y$ commute
for any $a, b$. (This is clear from writing out the matrix $S_b^y$ in block form
under the subspaces determined by the projections $\{ Q_a^x \}_x$.)

By Lemma~\ref{commutelemma}, we can find a unitary embedding $i \colon E \hookrightarrow E_1 \otimes E_2$ 
and POVMs $\{\overline{S}_b^y\}_y, \{\overline{Q}_a^x\}_x$ on $E_1, E_2$ such that
$S_b^y = i^* ( \overline{S}_b^y \otimes \mathbb{I} ) i$ and
$Q_a^x = i^* ( \mathbb{I} \otimes \overline{Q}_a^x ) i$.
With 
\begin{equation}
\overline{\gamma}= (\mathbb{I}_D \otimes i) \gamma (\mathbb{I}_D \otimes i^*),
\end{equation} 
the strategy
$\Gamma$ embeds into the strategy
\begin{eqnarray*}
 \Gamma':=\left( D, E_1 \otimes E_2, \{ \{ R_a^x \}_x \}_a , \{ \overline{S}_b^y \otimes \mathbb{I}_{E_2} \}_y \}_b,
\overline{\gamma}\right).
\end{eqnarray*}

For any fixed $a$, the state $\overline{\gamma}$ is such that applying
the measurement $\{ R_a^x \}_x$ to the system $D$ and the measurement
$\{ \overline{Q}_a^x \}_x$ to the system $E_2$ always yields the same outcome.
In particular, if we let
\begin{eqnarray}
\tau_a^x & = & \Tr_{E_2} \left( \overline{Q}_a^x \overline{\gamma}  \right),
\end{eqnarray}
then $\Tr [ R_a^{x'} \tau_a^x ]$ will always be equal to $1$ if $x = x'$ and equal to $0$ otherwise.
Therefore $\{ R_a^{x} \}_x$ commutes with the operators $\{ \tau_a^x \}_x$, and thus also with 
their sum $\sum_x \tau_a^x = \Tr_{E_2} \overline{\gamma}$. 

Thus
if we trace out the strategy $\Gamma'$ over the system $E_2$,
we obtain a strategy (congruent to the original strategy $\Gamma$)
in which Alice's measurement operators commute with the shared state.
\end{proof}

\section{Approximate Guessing}

Next we address the case where the second-player states
$\rho_{ab}^{xy}$ are not necessarily perfectly distinguishable
as $x$ varies, but are approximately distinguishable.  (Thus, 
if Bob were given Alice's input after the game was played and
asked to guess her output, he could do so with probability close to $1$.)
We begin by quantizing ``approximate'' distinguishability.

\begin{definition}
Let $\{ \rho_i \}_{i=1}^n$ denote a finite set of positive semidefinite
operators on a finite dimensional Hilbert space $V$.  Then, let 
\begin{eqnarray}
\Dist \{ \rho_i \} & = & \max \sum_i \Tr ( T_i \rho_i ), 
\end{eqnarray}
where the maximum is taken over all POVMs $\{ T_i \}_{i=1}^n$
on $V$.
\end{definition}

Note that if $\sum_i \Tr ( \rho_i) = 1$, and each $\rho_i$ is nonzero, then this quantity has
the following interpretation:
if Alice gives Bob a state from the set $\{ \rho_i / \Tr ( \rho_i ) \}$
at random according to the distribution $(\Tr (\rho_i ) )_i$, then $\Dist \{ \rho_i \}$
is the optimal probability that Bob can correctly guess the state.  This quantity
is well-studied (see, e.g., \cite{Spehner:2014_published}).

When we discussed perfect distinguishability, we made use 
of measurements that commuted with a given state.
In the current section we will need an approximate
version of such commutativity, and thus we make the following definition.
\begin{definition}
Let $\Phi \colon L ( V ) \to L ( V )$ denote a completely positive trace-preserving
map over a finite-dimensional Hilbert space $V$.  Let $\beta \in L ( V )$ denote a density
operator on $V$.  Then we say that \textbf{$\Phi$ is $\epsilon$-commutative
with $\beta$ if}
\begin{eqnarray}
\left\| \Phi ( \beta ) - \beta \right\|_1 & \leq & \epsilon.
\end{eqnarray}
\end{definition}
Note that this relation obeys a natural triangle inequality:
if $\Phi_1$ is $\epsilon_1$-commutative with $\beta$, and $\Phi_2$
is $\epsilon_2$-commutative with $\beta$, then
\begin{eqnarray*}
\left\| \Phi_2 ( \Phi_1 ( \beta )) - \beta \right\|_1 & \leq &
\left\| \Phi_2 ( \Phi_1 ( \beta ) ) - \Phi_2 ( \beta) \right\|_1 + \left\|
\Phi_2 ( \beta ) - \beta \right\|_1 \\
& \leq & \left\| \Phi_1 ( \beta) - \beta \right\|_1 + \epsilon_2 \\
& \leq & \epsilon_1 + \epsilon_2.
\end{eqnarray*}

The following known proposition will be an important building block.  We give a proof
that is a significant
simplification of a method from Lemma 29 in \cite{vidickthesis}.
(See also Lemma 2 in \cite{Mancinska:2014} for a related result.)

\begin{proposition}
\label{disturbprop}
Let $\Lambda \in L ( A \otimes B )$ be a density operator
and $\{ F_i \}_{i=1}^n$ a projective measurement on $A$ such that the induced states
$\Lambda_i^B := \Tr_A ( F_i \Lambda )$ satisfy
\begin{eqnarray}
\Dist \{ \Lambda^B_i \} & = & 1 - \delta.
\end{eqnarray}
Then, the superoperator $X \mapsto \sum_i F_i X F_i$ is
$(2 \sqrt{\delta} + \delta )$-commutative with $\Lambda^A := \Tr_B \Lambda$.
\end{proposition}

\proof{
Let $\alpha = \Lambda^A$.  
By assumption, there exists a POVM $\{G_i \}$ on $B$ such that
\begin{eqnarray}
\sum_i \Tr [ (F_i \otimes G_i ) \Lambda ] & = & 1 - \delta.
\end{eqnarray}
By standard arguments, we can assume without loss of generality
that $\{G_i \}$ is a projective measurement and that $\Lambda$ is pure.\footnote{
We can construct an enlargement
$B \subseteq \overline{B}$ such that $\mathbf{P}_B \overline{G}_i
\mathbf{P}_B = G_i$ for some projective measurement
$\{ \overline{G}_i \}$ on $\overline{B}$, and we can construct
an additional Hilbert space $E$ and a pure state $\overline{\Lambda}
\in L ( A \otimes B \otimes E )$ such that $\Tr_E \overline{\Lambda}
= \Lambda$.  The joint probability distribution of the measurements
$\{ F_i \}$ and $\{ \overline{G}_i \otimes I_E \}$ on $\overline{\Lambda}$
are the same as those of $\{ F_i \}$ and $\{ G_i \}$ on $\Lambda$.}

There is a linear map $M \colon \mathbb{C}^s \to \mathbb{C}^r$ 
such that $\Tr_A \Lambda = M^* M$ and $\alpha = \Tr_B \Lambda = M M^*$.
Upon choosing an appropriate basis for $A$ and $B$, we
can write $M$ with a block form determined by the spans of $\{F_i \}$ and $\{ G_j \}$:
\begin{eqnarray}
M & = & \left[ \begin{array}{c|c|c|c}
M_{11} & M_{12} & \cdots & M_{1n} \\
\hline
M_{21} & M_{22} & \cdots & M_{2n} \\
\hline
\vdots & & \ddots \\
\hline
M_{n1} & M_{n2} & \cdots & M_{nn}
\end{array} \right].
\end{eqnarray}
Let
\begin{eqnarray}
\overline{M} & = & \left[ \begin{array}{c|c|c|c}
M_{11} & 0 & \cdots & 0 \\
\hline
0 & M_{22} & \cdots & 0 \\
\hline
\vdots & & \ddots \\
\hline
0 & 0 & \cdots & M_{nn}
\end{array} \right].
\end{eqnarray}
Note that the probability of obtaining outcome $F_i$ for the measurement
on $A$ and outcome $G_j$ for the measurement on $B$ is
given by the quantity $\left\| M_{ij} \right\|_2^2$, and  the probability
that the outcomes of the measurements disagree is exactly $\left\| M -
\overline{M} \right\|_2^2$.  We have
\begin{eqnarray}
\left\| M - \overline{M} \right\|_2^2 & = & \delta.
\end{eqnarray}
Additionally, we can compare $\overline{M} \overline{M}^*$ to the post-measurement
state $\sum_i F_i \alpha F_i$.  The latter quantity is given by
\begin{eqnarray*}
\left[ \begin{array}{c|c|c|c}
\sum_k M_{1k} M_{1k}^* & 0 & \cdots & 0 \\
\hline
0 &  \sum_k M_{2k} M_{2k}^* & \cdots & 0 \\
\hline
\vdots & & \ddots \\
\hline
0 & 0 & \cdots & \sum_k M_{nk} M_{nk}^* \\
\end{array} \right], 
\end{eqnarray*}
and therefore the difference $(\sum_i F_i \alpha F_i  - \overline{M} \overline{M}^* )$
is equal to
\begin{eqnarray*}
\left[ \begin{array}{c|c|c|c}
\sum_{k\neq 1} M_{1k} M_{1k}^* & 0 & \cdots & 0 \\
\hline
0 &  \sum_{k \neq 2} M_{2k} M_{2k}^* & \cdots & 0 \\
\hline
\vdots & & \ddots \\
\hline
0 & 0 & \cdots & \sum_{k \neq n} M_{nk} M_{nk}^*, \\
\end{array} \right] \\
\end{eqnarray*}
which is a positive semidefinite operator
whose trace is exactly $\sum_{i \neq j} \left\| 
M_{ij} \right\|^2_2 = \delta$.  Thus,
\begin{eqnarray}
\left\| \sum_i F_i \alpha F_i - \overline{M} \overline{M}^* \right\|_1 
& = & \delta.
\end{eqnarray}
Therefore we have the following, using the Cauchy-Schwarz inequality:
\begin{eqnarray*}
&& \left\| \alpha - \sum_i F_i \alpha F_i \right\|_1 \\
& = &
\left\| M M^* - \sum_i F_i \alpha F_i \right\|_1 \\
& = & \left\| M ( M - \overline{M}^* ) + (M - \overline{M}) \overline{M}^* + \overline{M} \overline{M}^* - 
\sum_i F_i \alpha F_i \right\|_1 \\
& \leq & \left\| M ( M - \overline{M}^* ) \right\|_1 + \left\| (M - \overline{M}) \overline{M}^* \right\|_1 
\\ && + \left\| \overline{M} \overline{M}^* - 
\sum_i F_i \alpha F_i \right\|_1  \\
& \leq & \left\| M \right\|_2 \left\|  M - \overline{M}^*  \right\|_2 + \left\| M - \overline{M} \right\|_2 \left\| \overline{M}^* \right\|_2 + \delta \\
& \leq & 1 \cdot \sqrt{\delta} + \sqrt{\delta} \cdot 1 + \delta \\
& \leq & 2 \sqrt{\delta} + \delta,
\end{eqnarray*}
as desired.
}

The previous proposition showed that if a measurement by Alice is
highly predictable to Bob, then it does not disturb Alice's marginal state
by much.
The next corollary asserts Alice's measurement must also approximately preserve
any existing \textit{classical} correlation that Bob has with Alice's state.

\begin{corollary}
\label{classicalinfocor}
Let $\Lambda \in L ( A \otimes B \otimes C)$ be a density operator
which is classical on $C$.  (That is, $\Lambda = \sum_k \Lambda_k \otimes
\left| c_k \right> \left< c_k \right|$ for some orthonormal basis $\{c_1, \ldots, c_k \} 
\subseteq C$.)
Suppose that $\{ F_i \}_{i=1}^n$ is a projective measurement on $A$ such that the induced
states $\Lambda^{BC}_i := \Tr_A ( F_i \Lambda )$ satisfy
\begin{eqnarray}
\label{assumptionrepeat}
\Dist \{ \Lambda^{BC}_i \} & = & 1 - \delta.
\end{eqnarray}
Then, the superoperator $X \mapsto \sum_i (F_i \otimes I ) X (F_i \otimes I)$ is
$(2 \sqrt{\delta} + \delta )$-commutative with $\Lambda^{AC}$.
\end{corollary}

\proof{
Let $\overline{C}$ be a Hilbert space
which is isomorphic to $C$, and let $\overline{\Lambda} \in L ( A \otimes B \otimes C \otimes \overline{C})$
be the state that arises from $\Lambda$ by copying out along the standard
basis: $\left| c_i \right> \mapsto \left| c_i \overline{c_i } \right>$.  This copying
leaves the state $ABC$ unaffected, so assumption (\ref{assumptionrepeat}) still applies.
Thus by Proposition~\ref{disturbprop}, the operator
$X \mapsto \sum_i ( F_i \otimes I ) X (F_i \otimes I )$ is $(2 \sqrt{\delta} + \delta )$-commutative
with $\Lambda^{A \overline{C}}$, and the same
holds for the isomorphic state $\Lambda^{AC}$.
}

Now we prove a preliminary version of our main result.  We assume
that the states $\{ \rho_{ab}^{xy} \}_x$ are highly distinguishable
on average, and then deduce that Alice's and Bob's correlation
must be approximately classical.

\begin{proposition}
\label{uniformprop}
Let
\begin{eqnarray}
\Gamma & = & ( D, E, \{ \{ R_a^x \}_x \}_a , 
\{ \{ S_b^y \}_y \}_b , \gamma )
\end{eqnarray}
be a two-player strategy.
Let
\begin{eqnarray}
\delta & = & 1 - \frac{1}{| \mathcal{A} | |\mathcal{B } |} \sum_{aby} \Dist \{ \rho_{ab}^{xy} \mid x \in \mathcal{X} \}.
\end{eqnarray}
Then, there exists a classical correlation $(\overline{p}_{ab}^{xy} )$ such that
\begin{eqnarray}
\frac{1}{| \mathcal{A} | |\mathcal{B} | } \sum_{abxy} \left| p_{ab}^{xy} - \overline{p}_{ab}^{xy} \right| & \leq & 
\sqrt{ 3 \delta} \left| \mathcal{A} \right|.
\end{eqnarray}
\end{proposition}

\proof{
We can assume without loss of generality that the measurements $\{ \{ R_a^x \}_x \}_a$ are all projective.
We begin with the same strategy as in the proof of  Proposition~\ref{commutingprop}.
For each $a \in \mathcal{A}$, let $V_a = \mathbb{C}^\mathcal{X}$,
and let $\Phi_a \colon L ( D ) \to L ( V_a \otimes D )$
be the nondestructive measurement defined by 
\begin{eqnarray}
\Phi_a ( T ) & = & \sum_{x \in \mathcal{X}} \left| x \right> \left< x \right|
\otimes R_a^x T R_a^x.
\end{eqnarray}
Let $\Phi_a^{V_a} = \Tr_D \circ \Phi_a$ and let $\Phi_a^D = \Tr_{V_a} \circ
\Phi_a$. Likewise let $W_b = \mathbb{C}^\mathcal{Y}$ for each $b \in \mathcal{B}$, let
$\Psi_b \colon L ( E ) \to L ( W_b \otimes E )$ be the nondestructive measurement
defined by
\begin{eqnarray}
\Psi_b ( T ) & = & \sum_{y \in \mathcal{Y}} \left| y \right> \left< y \right|
\otimes \sqrt{S_b^y} T \sqrt{S_b^y}.
\end{eqnarray}
Let $\Psi_b^{W_b} = \Tr_E \circ \Psi_b$ and $\Psi_b^E = \Tr_{W_b} \circ \Psi_b$.

Assume without loss of generality that $\mathcal{A} = \{ 1, 2, \ldots, n \}$.
Let $\Lambda  \in L ( V_1 \otimes \ldots \otimes V_n \otimes D \otimes E )$
be the state that arises from applying the superoperators
$\Phi_1 \otimes I_E, \ldots, \Phi_n \otimes I_E$, in order, to $\gamma$.  Let
$(\overline{p}_{ab}^{xy} )$ be the correlation that arises
from Alice and Bob sharing the reduced state $\Lambda^{V_1 \ldots V_n E }$,
Alice obtaining her output on input $a$ from
the register $V_a$, and Bob obtaining his output from his prescribed measurements
$\{ \{S_b^y \}_y \}_b$ to $E$.  Since the state $\Lambda^{V_1 \ldots V_n E }$
is a separable state over the bipartition $(V_1 \ldots V_n \mid E )$, the correlation
$( \overline{p}_{ab}^{xy} )$ is classical.

Let
\begin{eqnarray}
\delta_{ab} & := &  1 - \sum_y \Dist \{ \rho_{ab}^{xy} \mid x \in \mathcal{X} \}.
\end{eqnarray}
If Alice and Bob share the measured state $(I_D \otimes \Psi_b ) (\gamma)$ partitioned
as $(D \mid E W_b)$, then the probability that
Bob can guess Alice's outcome when she measures with $\{ R_a^x \}_x$ is given by 
$(1 - \delta_{ab})$.
By Corollary~\ref{classicalinfocor}, the operator $(\Phi_a^D \otimes I_{W_b} )$
is $(2 \sqrt{\delta_{ab}} + \delta_{ab} )$-commutative with 
$(I_D \otimes \Psi_b^{W_b} ) \gamma$.

We wish to compare $(p_{ab}^{xy} )$ and $(\overline{p}_{ab}^{xy} )$.  
For any $a, b$ the probability
vector $(\overline{p}_{ab}^{xy})_{xy}$ describes the joint distribution
of the registers $V_a W_b$ under the density operator
\begin{eqnarray}
((\Phi_a^{V_a} \circ \Phi^D_{a-1} \circ \Phi^D_{a-2} \circ \cdots \circ \Phi^D_1 ) \otimes \Psi^{W_b}_b) \gamma,
\end{eqnarray}
which by the previous paragraph is within trace-distance $\sum_{i=1}^{a-1}
(2 \sqrt{\delta_{ab} } + \delta_{ab} )$ from the distribution described by
$(p_{ab}^{xy})_{xy}$:
\begin{eqnarray}
(\Phi_a^{V_a}  \otimes \Psi^{W_b}_b) \gamma.
\end{eqnarray}

Thus we have the following, in which we use the Cauchy-Schwarz inequality:
\begin{eqnarray}
\sum_{abxy} \left| p_{ab}^{xy} - \overline{p}_{ab}^{xy} \right| & \leq & 
\sum_{ab} \sum_{i=1}^{a-1} ( 2 \sqrt{\delta_{ib}} + \delta_{ib} ) \\
& = & \sum_{ab} (n-a) (2 \sqrt{\delta_{ab}} + \delta_{ab} ). \\
& \leq & \sum_{ab} (n-a) 3 \sqrt{\delta_{ab}} \\
& \leq & 3 \sqrt{ \sum_{ab} (n-a)^2 } \sqrt{ \sum_{ab} \delta_{ab}}  \\
& = & 3 \sqrt{ \sum_{ab} (n-a)^2 }  \sqrt{n | \mathcal{B} | \delta} \\
& = & 3 \sqrt{ | \mathcal{B}  | \sum_{a} (n-a)^2 } \sqrt{n | \mathcal{B} | \delta} \\
& = & 3 | \mathcal{B}  | \sqrt{ \sum_{a} (n-a)^2 } \sqrt{n \delta} \\
& \leq & 3 | \mathcal{B}| \sqrt{ n^3/3 } \sqrt{n \delta},
\end{eqnarray}
which simplifies to the desired bound.
}

Proposition~\ref{uniformprop} is useful for addressing any game $(q, H)$
where the distribution $q$ is uniform (i.e., $q(a, b) = 1/(|\mathcal{A} | | \mathcal{B} | )$.)  We prove the following theorem
which applies to more general games.

\begin{theorem}
\label{robustthm}
Let $G = (q, H)$ be a complete-support game and let 
\begin{eqnarray}
\Gamma & = & ( D, E, \{ \{ R_a^x \}_x \}_a , 
\{ \{ S_b^y \}_y \}_b , \gamma )
\end{eqnarray}
be a two-player strategy.
Let
\begin{eqnarray}
\epsilon & = & 1 - \sum_{ab} q(a, b)  \sum_y \Dist \{ \rho_{ab}^{xy} \mid x \in \mathcal{X} \}.
\end{eqnarray}
Then, the score achieved by $\Gamma$ exceeds the best classical score $\omega_c (G )$ by at most
$C_G \sqrt{\epsilon}$, where
\begin{eqnarray}
\label{cg}
C_G & = & (3/2) \sqrt{ \sum_{ab} q ( b ) \left( q ( a \mid b ) \right)^{-1}}.
\end{eqnarray}
\end{theorem}

\proof{
Define $\overline{p}_{ab}^{xy}$ and $\delta_{ab}$ as in Proposition~\ref{uniformprop}.  We have the following (again
using the Cauchy-Schwartz inequality):
\begin{eqnarray}
&& \sum_{abxy} q ( a, b ) | p_{ab}^{xy} - \overline{p}_{ab}^{xy} |
\\
& \leq & \sum_{ab} q ( a, b ) \sum_{i=1}^{a-1} (2 \sqrt{\delta_{ib}} + \delta_{ib} ) \\
& \leq & \sum_{ab} q ( a, b ) \sum_{i=1}^{a-1} 3 \sqrt{\delta_{ib}}   \\
& = & \sum_{ab} \left( \sum_{k = a+1}^n q ( k, b ) \right) 3 \sqrt{\delta_{ab}} \\
& = & \sum_{ab} \left( \frac{\sum_{k=a+1}^n q ( k , b ) }{\sqrt{q ( a, b ) } }\right) 3 \sqrt{ q ( a, b ) \delta_{ab}} \\
& \leq & 3 \sqrt{ \sum_{ab} \frac{ (\sum_{k=a+1}^n q ( k , b ) )^2}{q ( a, b ) }} \sqrt{ \sum_{ab} q ( a, b ) \delta_{ab}} \\
& \leq & 3 \sqrt{ \sum_{ab} \frac{ (\sum_{k=a+1}^n q ( k , b ) )^2}{q ( a, b ) }} \sqrt{ \epsilon } \\
& \leq & 3 \sqrt{ \sum_{ab} \frac{ q ( b )^2}{q ( a, b ) }}\sqrt{\epsilon} \\
& \leq & 2 C_G \sqrt{ \epsilon} \label{finalbound}
\end{eqnarray}
Note that for any probability vectors $\mathbf{t} = (t_1, \ldots, t_m)$ and $\mathbf{s} = (s_1, \ldots, s_m)$ and any arbitrary vector
$(u_1, \ldots, u_m) \in [0, 1]^m$, we have
\begin{eqnarray}
\sum_i u_i ( t_i - s_i ) \leq \frac{1}{2} \sum \left| t_i - s_i \right|.
\end{eqnarray}
Applying this fact to the probability vectors $(q(a,b)p_{ab}^{xy})_{abxy}$
and $(q(a,b)\overline{p}_{ab}^{xy} )_{abxy}$  and the 
vector $(H ( a, b, x, y ) )_{abxy}$ implies that the difference
between the score achieved by $(p_{ab}^{xy})$ and the score
achieved by $(\overline{p}_{ab}^{xy})$ is no more than half the
quantity (\ref{finalbound}), which yields the desired result.
}

\section{Discussion}
\label{sec:disc}

When two players achieve a superclassical score at a nonlocal game, their
outputs must be at least partially unpredictable to an outside party, even
if that party knows the inputs that were given.  This fact
is one of the bases for randomness expansion from untrusted devices
\cite{col:2006}, where a user referees a nonlocal game repeatedly with $2$
or more untrusted players (or, equivalently, $2$ or more untrusted quantum devices)
to expand a small uniformly random seed $S$
into a large output string $T$ that is uniform conditioned on $S$.  The players can exhibit arbitrary
quantum behavior, but it is assumed that they are prevented from communicating
with the adversary.  At the center of some of the
discussions of randomness expansion (e.g., \cite{pam:2010}) is the
fact that the min-entropy of the outputs of the players can be lower bounded
by an increasing function of the score achieved at the game.

In this paper we have proven an analogous result for the case where
one player in a game wishes to generate randomness that is unknown
to the other player --- in other words, we have achieved (one-shot) \textit{blind}
randomness expansion.  (The second party, Bob, is ``blind'' to the randomness
generated by Alice.)  We have also proven a general rate curve for any game $G$,
which relates the score achieve at $G$ to the predictability of Alice's output from the perspective of Bob --
specifically, if $G$ is a complete support game and Alice and Bob achieve score $w$, then
Bob's probability of guessing her output given her input is at most
\begin{eqnarray*}
f_G ( w ) & = & \left\{ \begin{array}{cl}
1 -( w - \omega_c ( G ) )^2/C_G & \textnormal{ if } w \geq \omega_c ( G ) \\
1 & \textnormal{ otherwise,} \end{array} \right.
\end{eqnarray*}
where $C_G$ denotes the constant defined in equation (\ref{cg}).  

A possible next step would be to prove a multi-shot version
of Theorem~\ref{robustthm}, e.g., a proof that Alice's outputs across multiple rounds
have high smooth min-entropy from Bob's perspective.  With the use of a
quantum-proof randomness extractor (e.g., \cite{De:2012}) this would imply that
Alice has the ability to generate uniformly random bits, known only to her, through interactions with Bob.
In the device-independent setting, this would mean that one device could be be reused
in multiple iterations  of randomness expansion without affecting the security guarantee,
and in particular would decrease the minimum number of quantum devices needed to perform unbounded
randomness expansion from four (as in \cite{Miller:2016, CSW14}) down to three.

The recent entropy accumulation theorem
\cite{Dupuis:2016} proves lower bounds on smooth min-entropy in various scenarios where
a Bell inequality is violated.  It will be interesting
to see if it can be generalized to cover blind randomness expansion as well.
(The current results apply under a Markov assumption 
 which is not satisfied in our case.)

A corollary of our result is that, for any complete-support game $G$,
the range of scores that certify randomness against a third party are exactly the same
as the range of scores that certify randomness for one player against the second --- in both cases,
any superclassical score is adequate.  We point out, however, that the certified min-entropy can be different.
A simple example of this is the Magic Square game, where Alice and Bob are given inputs $a, b \in \{ 1, 2, 3 \}$
respectively, and must produce outputs $(x_1, x_2, x_3), (y_1, y_2, y_3) \in \{ 0, 1 \}^3$ respectively which satisfy
\begin{eqnarray}
x_1 \oplus x_2 \oplus x_3 & = & 0, \\
y_1 \oplus y_2 \oplus y_3 & = & 1, \\
x_b & = & y_a.
\end{eqnarray}
Self-testing \cite{WuBancal:2016_p} for the Magic Square
game implies that if Alice and Bob achieve a perfect score, Alice's output contains two bits of perfect randomness
from the perspective of a third party, but only one perfect bit of randomness from the perspective of Bob. 
Optimizing the relationship between the game score and min-entropy in the blind scenario is an open problem.

A potentially useful aspect of Corollary~\ref{maincor} is that it contains
a notion of \textit{certified erasure} of information.  For the example of the Magic Square game
mentioned above, if Bob were asked before his turn to guess Alice's output
given her input, he could do this perfectly.  (The optimal strategy for the Magic Square
game uses a maximally entangled state and projective measurements, so each party's measurement
outcomes can be perfectly guessed by the other player.)
Contrary to this, when Bob is compelled to carry out his part of the
strategy before Alice's input is revealed, he loses the ability
to perfectly guess Alice's output.  Requiring a superclassical score from Alice and Bob
amounts to forcing Bob to erase information.  Different
variants of certified erasure are a topic of current study
\cite{Unruh:2015, Kaniewski:2016, Rib:2016}.
An interesting
research avenue is to determine the minimal assumptions under
which certified erasure is possible.

Finally, we note that the scenario in which the second player
tries to guess the first player's output after computing his own
output fits the general framework of \textit{sequential
nonlocal correlations} \cite{Gallego:2014}.  \commentout{ {\color{red} [Cite 
Teufel 1997?]} }
In \cite{Curchod:2015} such correlations
are used for ordinary (non-blind) randomness expansion.
A next step is to explore
how our techniques could be applied to more general
sequential nonlocal games.

\paragraph{Acknowledgments.} 
We are indebted to Laura Mancinska for discussions that helped
us to prove our robust result (Theorem~\ref{robustthm}).
The first author also thanks Jedrzej Kaniewski, Marcin Pawlowski and Stefano Pironio  for
helpful information. This research was supported in part by US NSF
Awards 1500095, 1216729, 1526928, and 1318070.

\bibliographystyle{apsrev}
\bibliography{quantumsec}

\end{document}